\documentclass[letterpaper]{article}
\usepackage{fullpage}
\usepackage{amsmath, amssymb,amsfonts,amsthm}
\newtheorem{theorem}{Theorem}
\newtheorem{lemma}{Lemma}

\newtheorem{definition}{Definition}

\newcommand{\rdelta}{\textbf{d}}

\newcommand{\V}{V}
\newcommand{\E}{\mathbb{E}}
\newcommand{\is}{{i^*}}

\newcommand{\vecspan}{\textrm{span}}
\newcommand{\VecSet}{X}
\newcommand{\Matrix}{A}
\newcommand{\uv}{\textbf{u}}

\newcommand{\y}{\textbf{y}}
\newcommand{\U}{\textbf{U}}
\newcommand{\inc}[1]{\chi_{#1}}
\newcommand{\z}{\textbf{z}}
\newcommand{\vs}{\textbf{v}}

\newcommand{\Q}{\textbf{Q}}
\newcommand{\q}{\textbf{q}}
\newcommand{\eps}{\varepsilon}

\begin{document}

\title{The Cell Probe Complexity of Dynamic Range Counting}


\author{
Kasper Green Larsen\thanks{Kasper Green Larsen is a recipient
    of the Google Europe Fellowship in Search and Information
    Retrieval, and this research is also supported in part by this
    Google Fellowship.}\\
       MADALGO\thanks{Center for Massive Data Algorithmics, a Center of the Danish National Research Foundation.}, Department of Computer Science\\
       Aarhus University\\
       Aarhus, Denmark\\
       email: larsen@cs.au.dk
}

\maketitle

\begin{abstract} 
In this paper we develop a new technique for proving lower bounds on
the update time and query time of dynamic data structures in the cell
probe model. With this technique, we prove the highest lower bound to
date for any explicit problem, namely a lower bound of
$t_q=\Omega((\lg n/\lg(wt_u))^2)$. Here $n$ is the number of update
operations, $w$ the cell size, $t_q$ the query time and $t_u$ the
update time. In the most natural setting of cell size $w=\Theta(\lg
n)$, this gives a lower bound of $t_q=\Omega((\lg n/\lg \lg n)^2)$ for
any polylogarithmic update time. This bound is almost a quadratic
improvement over the highest previous lower bound of $\Omega(\lg n)$,
due to P{\v a}tra{\c s}cu and Demaine [SICOMP'06].

We prove the lower bound for the fundamental problem of weighted
orthogonal range counting. In this problem, we are to support
insertions of two-dimensional points, each assigned a $\Theta(\lg
n)$-bit integer weight. A query to this problem is specified by a
point $q=(x,y)$, and the goal is to report the sum of the weights
assigned to the points dominated by $q$, where a point $(x',y')$ is
dominated by $q$ if $x' \leq x$ and $y' \leq y$. In addition to being
the highest cell probe lower bound to date, the lower bound is also
tight for data structures with update time $t_u =
\Omega(\lg^{2+\eps}n)$, where $\eps>0$ is an arbitrarily small
constant.
\end{abstract}

\newpage

\section{Introduction}
Proving lower bounds on the operational time of data structures has
been an active line of research for decades. During these years,
numerous models of computation have been proposed, including the cell
probe model of Yao~\cite{yao:cellprobe}. The cell probe model is one
of the least restrictive lower bound models, thus lower bounds proved
in the cell probe model apply to essentially every imaginable data
structure, including those developed in the popular upper bound model,
the word RAM. Unfortunately this generality comes at a cost: The
highest lower bound that has been proved for any explicit data
structure problem is $\Omega(\lg n)$, both for static and even dynamic
data structures\footnote{This is true under the most natural
  assumption of cell size $\Theta(\lg n)$.}.

In this paper, we break this barrier by introducing a new technique
for proving dynamic cell probe lower bounds. Using this technique, we
obtain a query time lower bound of $\Omega((\lg n/\lg \lg n)^2)$ for
any polylogarithmic update time. We prove the bound for the
fundamental problem of dynamic \emph{weighted orthogonal range
  counting} in two-dimensional space. In dynamic weighted orthogonal
range counting (in 2-d), the goal is to maintain a set of (2-d) points
under insertions, where each point is assigned an integer weight. In
addition to supporting insertions, a data structure must support
answering queries. A query is specified by a query point $q=(x,y)$,
and the data structure must return the sum of the weights assigned to
the points \emph{dominated} by $q$. Here we say that a point
$(x',y')$ is dominated by $q$ if $x' \leq x$ and $y' \leq y$.

\subsection{The Cell Probe Model}
A dynamic data structure in the cell probe model consists of a set of
memory cells, each storing $w$ bits. Each cell of the data structure
is identified by an integer address, which is assumed to fit in $w$
bits, i.e. each address is amongst $[2^w]=\{0,\dots,2^w-1\}$. We will
make the additional standard assumption that a cell also has enough
bits to address any update operation performed on it, i.e. we assume
$w = \Omega(\lg n)$ when analysing a data structure's performance on a
sequence of $n$ updates.

When presented with an update operation, a data structure reads and
updates a number of the stored cells to reflect the changes. The cell
read (or written to) in each step of an update operation may depend
arbitrarily on the update and the contents of all cells previously
probed during the update. We refer to the reading or writing of a cell
as probing the cell, hence the name cell probe model. The update time
of a data structure is defined as the number of cells probed when
processing an update.

To answer a query, a data structure similarly probes a number of cells
from the data structure and from the contents of the probed cells,
the data structure must return the correct answer to the query. Again,
the cell probed at each step, and the answer returned, may be an
arbitrary function of the query and the previously probed cells. We
similarly define the query time of a data structure as the number of
cells probed when answering a query.

\paragraph{Previous Results.}
In the following, we give a brief overview of the most important
techniques that have been introduced for proving cell probe lower
bounds for dynamic data structures. We also review the previous cell
probe lower bounds obtained for orthogonal range counting and related
problems. In Section~\ref{sec:tech} we then give a more thorough
review of the previous techniques most relevant to this work, followed
by a description of the key ideas in our new technique.

In their seminal paper~\cite{Fredman:chrono}, Fredman and Saks
introduced the celebrated chronogram technique. They applied their
technique to the \emph{partial sums} problem and obtained a lower
bound stating that $t_q = \Omega(\lg n/\lg(wt_u))$, where $t_q$ is the
query time and $t_u$ the update time. In the partial sums problem, we
are to maintain an array of $n$ $O(w)$-bit integers under updates
of the entries. A query to the problem consists of two indices $i$ and
$j$, and the goal is to compute the sum of the integers in the
subarray from index $i$ to $j$. The lower bound of Fredman and Saks
holds even when the data structure is allowed amortization and
randomization.

The bounds of Fredman and Saks remained the highest achieved until the
breakthrough results of P{\v a}tra{\c s}cu and
Demaine~\cite{Patrascu:loga}. In their paper, they extended upon the
ideas of Fredman and Saks to give a tight lower bound for the partial
sums problem. Their results state that $t_q\lg(t_u/t_q)=\Omega(\lg n)$
and $t_u \lg(t_q/t_u)=\Omega(\lg n)$ when the integers have
$\Omega(w)$ bits, which in particular implies
$\max\{t_q,t_u\}=\Omega(\lg n)$. We note that they also obtain tight
lower bounds in the regime of smaller integers. Again, the bounds hold
even when allowed amortization and randomization. For the most natural
cell size of $w=\Theta(\lg n)$, this remains until today the highest
achieved lower bound.

The two above techniques both lead to smooth tradeoff curves between
update time and query time. While this behaviour is correct for the
partial sums problem, there are many examples where this is certainly
not the case. P{\v a}tra{\c s}cu and Thorup~\cite{patrascu11unions}
recently presented a new extension of the chronogram technique, which
can prove strong threshold lower bounds. In particular they showed
that if a data structure for maintaining the connectivity of a graph
under edge insertions and deletions has amortized update
time just $o(\lg n)$, then the query time explodes to $n^{1-o(1)}$.

In the search for super-logarithmic lower bounds, P{\v a}tra{\c s}cu introduced
a dynamic set-disjointness problem named the multiphase
problem~\cite{patrascu10mp-3sum}. Based on a widely believed
conjecture about the hardness of 3-SUM, P{\v a}tra{\c s}cu first reduced 3-SUM
to the multiphase problem and then gave a series of reductions to
different dynamic data structure problems, implying polynomial lower
bounds under the 3-SUM conjecture.

Finally, we mention that P{\v a}tra{\c s}cu~\cite{patrascu07sum2D}
presented a technique capable of proving a lower bound of
$\max\{t_q,t_u\}=\Omega((\lg n/\lg \lg n)^2)$ for dynamic weighted
orthogonal range counting, but only when the weights are
$\lg^{2+\eps}n$-bit integers where $\eps>0$ is an arbitrarily small
constant. For range counting with $\delta$-bit weights, it is most
natural to assume that the cells have enough bits to store the
weights, since otherwise one immediately obtains an update time lower
bound of $\delta/w$ just for writing down the change. Hence his proof
is meaningful only in the case of $w=\lg^{2+\eps}n$ as well (as
he also notes). Thus the magnitude of the lower bound compared to the
number of bits, $\delta$, needed to describe an update operation (or a
query), remains below $\Omega(\delta)$.  This bound holds when $t_u$
is the worst case update time and $t_q$ the expected
average\footnote{i.e. for any data structure with a possibly
  randomized query algorithm, there exists a sequence of updates $U$,
  such that the expected cost of answering a uniform random query
  after the updates $U$ is $t_q$.} query time of a data structure.

The particular problem of orthogonal range counting has received much
attention from a lower bound perspective. In the static case, P{\v
  a}tra{\c s}cu~\cite{patrascu07sum2D} first proved a lower bound of
$t=\Omega(\lg n/\lg(Sw/n))$ where $t$ is the expected average query
time and $S$ the space of the data structure in number of cells. This
lower bound holds for regular counting (without weights), and even
when just the parity of the number of points in the range is to be
returned. In~\cite{patrascu08structures} he reproved this bound using
an elegant reduction from the communication game known as lop-sided
set disjointness. Subsequently J\o rgensen and
Larsen~\cite{larsen:median} proved a matching bound for the strongly
related problems of range selection and range median. Finally, as
mentioned earlier, P{\v a}tra{\c s}cu~\cite{patrascu07sum2D} proved a
$\max\{t_q,t_u\}=\Omega((\lg n/\lg \lg n)^2)$ lower bound for dynamic
weighted orthogonal range counting when the weights are
$\lg^{2+\eps}n$-bit integers. In the concluding remarks of that
paper, he posed it as an interesting open problem to prove the same
lower bound for regular counting.

\paragraph{Our Results.}
In this paper we introduce a new technique for proving dynamic cell
probe lower bounds. Using this technique, we obtain a lower bound of
$t_q=\Omega((\lg n/\lg(wt_u))^2)$, where $t_u$ is the worst case
update time and $t_q$ is the expected average query time of the data
structure. The lower bound holds for any cell size $w=\Omega(\lg n)$,
and is the highest achieved to date in the most natural setting of
cell size $w=\Theta(\lg n)$.  For polylogarithmic $t_u$ and
logarithmic cell size, this bound is $t_q = \Omega((\lg n/\lg \lg
n)^2)$, i.e. almost a quadratic improvement over the highest previous
lower bound of P{\v a}tra{\c s}cu and Demaine.

We prove the lower bound for dynamic weighted orthogonal range
counting in two-dimensional space, where the weights are $\Theta(\lg
n)$-bit integers. This gives a partial answer to the open problem
posed by P{\v a}tra{\c s}cu by reducing the requirement of the
magnitude of weights from $\lg^{2+\eps}n$ to just
logarithmic. Finally, the lower bound is also tight for any update
time that is at least $\lg^{2+\eps}n$, hence deepening our
understanding of one of the most fundamental range searching problems.

\paragraph{Overview.}
In Section~\ref{sec:tech} we discuss the two previous techniques most
related to ours, i.e. that of Fredman and Saks~\cite{Fredman:chrono}
and of P{\v a}tra{\c s}cu~\cite{patrascu07sum2D}. Following this
discussion, we give a description of the key ideas behind our new
technique. Having introduced our technique, we first demonstrate it on
an artificial range counting problem that is tailored for our
technique (Section~\ref{sec:ex}) and then proceed to the main lower
bound proof in Section~\ref{sec:dyn}. Finally we conclude in
Section~\ref{sec:conclude} with a discussion of the limitations of our
technique and the intriguing open problems these limitations pose.

\section{Techniques}
\label{sec:tech}
In this section, we first review the two previous techniques most
important to this work, and then present our new technique.

\paragraph{Fredman and Saks~\cite{Fredman:chrono}.} This technique is known as the chronogram technique. The basic idea is to consider batches, or \emph{epochs}, of updates to a data structure problem. More formally, one defines an epoch $i$ for each $i=1,\dots,\lg_\beta n$, where $\beta>1$ is a parameter. The $i$'th epoch consists of performing $\beta^i$ randomly chosen updates. The epochs occur in time from largest to smallest epoch, and at the end of epoch $1$, every cell of the constructed data structure is associated to the epoch in which it was last updated. The goal is to argue that to answer a query after epoch $1$, the query algorithm has to probe one cell associated to each epoch. Since a cell is only associated to one epoch, this gives a total query time lower bound of $\Omega(\lg_\beta n)$.

Arguing that the query algorithm must probe one cell associated to
each epoch is done by setting $\beta$ somewhat larger than the worst
case update time $t_u$ and the cell size $w$. Since cells associated
to an epoch $j$ cannot contain useful information about an epoch $i<j$
(the updates of epoch $j$ were performed before knowing what the
updates of epoch $i$ was), one can ignore cells associated to such
epochs when analysing the probes to an epoch $i$. Similarly, since all
epochs following epoch $i$ (future updates) writes a total of
$O(\beta^{i-1}t_u)=o(\beta^i)$ cells, these cells do not contain
enough information about the $\beta^i$ updates of epoch $i$ to be of
any use (recall the updates are random, thus there is still much
randomness left in epoch $i$ after seeing the cells written in epochs
$j<i$). Thus if the answer to a query depends on an update of epoch
$i$, then the query algorithm must probe a cell associated to epoch
$i$ to answer the query.

We note that Fredman and Saks also defined the notion of epochs over a
sequence of intermixed updates and queries. Here the epochs are
defined relative to each query, and from this approach they obtain
their amortized bounds.

\paragraph{P{\v a}tra{\c s}cu~\cite{patrascu07sum2D}.}
This technique uses the same setup as the chronogram technique,
i.e. one considers epochs $i=1,\dots,\lg_\beta n$ of updates, followed
by one query. The idea is to use a static $\Omega(\lg_\beta n)$ lower
bound proof to argue that the query algorithm must probe
$\Omega(\lg_\beta n)$ cells from each epoch if the update time is
$o((\lg n/\lg \lg n)^2)$, and not just one cell. Summing over all
epochs, this gives a lower bound of $\Omega(\lg_\beta^2 n)$. In the
following, we give a coarse overview of the general framework for
doing so.

One first proves a lower bound on the amount of communication in the
following (static) communication game (for every epoch $i$): Bob
receives all epochs of updates to the dynamic data structure problem
and Alice receives a set of queries and all updates of the epochs
preceding epoch $i$. The goal for them is to compute the answer to
Alice's queries after all the epochs of updates.

When such a lower bound has been established, one considers each epoch
$i$ in turn and uses the dynamic data structure to obtain an efficient
protocol for the above communication game between Alice and Bob. The
key idea is to let Alice simulate the query algorithm of the dynamic
data structure on each of her queries, and whenever a cell associated
to epoch $i$ is requested, she asks Bob for the contents. Bob replies
and she continues the simulation. Clearly the amount of communication
is proportional to the number of probes to cells associated to epoch
$i$, and thus a lower bound follows from the communication game lower
bound. The main difficulty in implementing this protocol is that Alice
must somehow recover the contents of the cells not associated to epoch
$i$ without asking Bob for it. This is accomplished by first letting
Bob send all cells associated to epochs $j<i$ to Alice. For
sufficiently large $\beta$, this does not break the communication
lower bound. To let Alice know which cells that belong to epoch $i$,
Bob also sends a Bloom filter specifying the addresses of the cells
associated to epoch $i$. A Bloom filter is a membership data
structure with a false positive probability. By setting the false
positive probability to $1/\lg^c n$ for a large enough constant $c>0$,
the Bloom filter can be send using $O(\lg \lg n)$ bits per cell
associated to epoch $i$. If $t_u = o((\lg n/\lg \lg n)^2)$, this
totals $o(\beta^i \lg^2n/\lg\lg n)$ bits.

Now Alice can execute the updates of the epochs preceding epoch $i$
(epochs $j>i$) herself, and she knows the cells (contents and
addresses) associated to epochs $j<i$. She also has a Bloom filter
specifying the addresses of the cells associated to epoch $i$. Thus to
answer her queries, she starts simulating the query algorithm. Each
time a cell is requested, she first checks if it is associated to
epochs $j<i$. If so, she has the contents herself and can continue the
simulation. If not, she checks the Bloom filter to determine whether
it belongs to epoch $i$. If the Bloom filter says no, the contents of
the cell was not updated during epochs $j \leq i$ and thus she has the
contents from the updates she executed initially. Finally, if the
Bloom filter says yes, she asks Bob for the contents. Clearly the
amount of communication is proportional to the number of probes to
cells associated to epoch $i$ plus some additional communication due
to the $t_q/\lg^cn$ false positives.

To get any lower bound out of this protocol, sending the Bloom filter
must cost less bits than it takes to describe the updates of epoch $i$
(Bob's input). This is precisely why the lower bound of P{\v a}tra{\c
  s}cu requires large weights assigned to the input points.

\paragraph{Our Technique.}
Our new technique elegantly circumvents the limitations of P{\v
  a}tra{\c s}cu's technique by exploiting recent ideas by Panigrahy et
al.~\cite{pani:metric} for proving static lower bounds. The basic
setup is the same, i.e. we consider epochs $i=1,\dots,\lg_\beta n$,
where the $i$'th epoch consists of $\beta^i$ updates. As with the two
previous techniques, we associate a cell to the epoch in which it was
last updates. Lower bounds now follow by showing that any data
structure must probe $\Omega(\lg_\beta n)$ cells associated to each
epoch $i$ when answering a query at the end of epoch $1$. Summing over
all $\lg_\beta n$ epochs, this gives us a lower bound of
$\Omega(\lg_\beta^2n)$.

To show that $\Omega(\lg_\beta n)$ probes to cells associated to an
epoch $i$ are required, we assume for contradiction that a data
structure probing $o(\lg_\beta n)$ cells associated to epoch $i$
exists. Using this data structure, we then consider a game between an
encoder and a decoder. The encoder receives as input the updates of
all epochs, and must from this send a message to the decoder. The
decoder then sees this message and all updates preceding epoch $i$
and must from this uniquely recover the updates of epoch $i$. If the
message is smaller than the entropy of the updates of epoch $i$ (conditioned on preceding epochs), this
gives an information theoretic contradiction. The trick is to find a
way for the encoder to exploit the small number of probed cells to
send a short message. 

As mentioned, we use the ideas in~\cite{pani:metric} to exploit the
small number of probes. In~\cite{pani:metric} it was observed that if
$S$ is a set of cells, and if the query algorithm of a data structure
probes $o(\lg_\beta n)$ cells from $S$ on average over all queries
(for large enough $\beta$), then there is a subset of cells $S'
\subseteq S$ which \emph{resolves} a large number of queries. Here we
say that a subset of cells $S' \subseteq S$ resolves a query, if the
query algorithm probes no cells in $S \setminus S'$ when answering
that query. What this observation gives us compared to the approach of
P{\v a}tra{\c s}cu, is that we can find a large set of queries that
are all resolved by the same small subset of cells associated to an
epoch $i$. Thus we no longer have to specify all cells associated to
epoch $i$, but only a small fraction.

With this observation in mind, the encoder proceeds as follows: First
he executes all the updates of all epochs on the claimed data
structure. He then sends all cells associated to epochs $j<i$. For
large enough $\beta$, this message is smaller than the entropy of the
$\beta^i$ updates of epoch $i$. Letting $S_i$ denote the cells
associated to epoch $i$, the encoder then finds a subset of cells
$S_i' \subseteq S_i$, such that a large number of queries are resolved
by $S_i'$. He then sends a description of those cells and proceeds by
finding a subset $Q$ of the queries resolved by $S_i'$, such that
knowing the answer to all queries in $Q$ reduces the entropy of the
updates of epoch $i$ by more than the number of bits needed to
describe $S_i',Q$ and the cells associated to epochs $j<i$. He then
sends a description of $Q$ followed by an encoding of the updates of
epoch $i$, conditioned on the answers to queries in $Q$. Since the
entropy of the updates of epoch $i$ is reduced by more bits than was
already send, this gives our contradiction (if the decoder can recover
the updates from the above messages).

To recover the updates of epoch $i$, the decoder first executes the
updates preceding epoch $i$. His goal is to simulate the query
algorithm for every query in $Q$ to recover all the answers. He
achieves this in the following way: For each cell $c$ requested when
answering a query $q \in Q$, he examines the cells associated to
epochs $j<i$ (those cells were send by the encoder), and if $c$ is
contained in one of those he immediately recovers the contents. If
not, he proceeds by examining the set $S_i'$. If $c$ is included in
this set, he has again recovered the contents and can continue the
simulation. Finally, if $c$ is not in $S_i'$, then $c$ must be
associated to an epoch preceding epoch $i$ (since queries in $Q$ probe
no cells in $S_i \setminus S_i'$), thus the decoder recovers the
contents of $c$ from the updates that he executed initially. In this
manner, the decoder can recover the answer to every query in $Q$, and
from the last part of the message he recovers the updates of epoch
$i$. 

The main technical challenge in using our technique lies in arguing
that if $o(\lg_\beta n)$ cells are probed amongst the cells associated
to epoch $i$, then the claimed cell set $S_i'$ and query set $Q$
exists.

In Section~\ref{sec:ex} we first use our technique to prove a lower
bound of $t_q = \Omega((\lg n/\lg(wt_u))^2)$ for an artificially
constructed range counting problem. This problem is tailored towards
giving as clean an introduction of our technique as possible. In
Section~\ref{sec:dyn} we then prove the main result, i.e. a lower
bound for dynamic weighted orthogonal range counting.

\section{An Artificial Range Counting Problem}
\label{sec:ex}
In the following, we design a range counting problem where the
queries have some very desirable properties. These properties ease the
lower bound proof significantly. We first describe the queries and
then give some intuition on why their properties ease the proof. The
queries are defined using the following lemma:

\begin{lemma}
\label{thm:hardqueries}
For $n$ sufficiently large and any prime $\Delta$, where $n^4/2 \leq
\Delta \leq n^4$, there exists a set $V$ of $n^2$ $\{0,1\}$-vectors in
$[\Delta]^n$, such that for any $\sqrt{n} \leq k \leq n$, it holds
that if we consider only the last $k$ coordinates of the vectors in
$V$, then any subset of up to $k/22\lg k$ vectors in $V$ are linearly
independent in $[\Delta]^k$.
\end{lemma}

We defer the (trivial) proof a bit and instead describe the artificial
range counting problem:

\paragraph{The Problem.}
For $n$ sufficiently large and any prime $\Delta$, where $n^4/2 \leq
\Delta \leq n^4$, let $V=\{v_0,\dots,v_{n^2-1}\}$ be a set of $n^2$
$\{0,1\}$-vectors with the properties of
Lemma~\ref{thm:hardqueries}. The set $V$ naturally defines a range
counting problem over a set of $n$ points: Each of $n$ input points
$p_0,\dots,p_{n-1}$ are assigned an integer weight amongst
$[\Delta]$. A query is specified by a vector $v_j$ in $V$ (or simply
an index $j \in [n^2]$) and the answer to the query is the sum of the
weights of those points $p_i$ for which the $i$'th coordinate of $v_j$
is $1$. An update is specified by an index $i \in [n]$ and an integer
weight $\delta \in [\Delta]$, and the effect of the update is to
change the weight of point $p_i$ to $\delta$. Initially, all weights
are $0$.

Observe that this range counting problem has $n^2$ queries and the
weights of points fit in $\lg \Delta \leq 4\lg n$ bits. Thus the
problem is similar in flavor to weighted orthogonal range
counting. Lemma~\ref{thm:hardqueries} essentially tells us that the
answers to any subset of queries reveal a lot information about the
weights of the $n$ input points (by the independence). Requiring that
the independence holds even when considering only the last $k$
coordinates is exploited in the lower bound proof to argue that the
answers to the queries reveal much information about any large enough
epoch $i$. Finally, recall from Section~\ref{sec:tech} that our new
technique requires us to encode a set of queries to simulate the query
algorithm for. Since encoding a query takes $\lg(|V|)=2\lg n$ bits, we
have chosen the weights to be $4\lg n$-bit integers, i.e. if we can
answer a query, we get more bits out than it takes to write down the
query.

\paragraph{Proof of Lemma~\ref{thm:hardqueries}.}
We prove this by a probabilistic argument. Let $n$ and $\Delta$ be
given, where $n^4/2 \leq \Delta \leq n^4$. Initialize $V$ to the empty
set. Clearly, for any $\sqrt{n} \leq k \leq n$, it holds that any
subset of up to $k/22\lg k$ vectors in $V$ are linearly independent in
$[\Delta]^k$ when considering only the last $k$ coordinates. We prove
that as long as $|V| < n^2$, we can find a $\{0,1\}$-vector whose
addition to the set $V$ maintains this property. For this, consider a
uniform random vector $v$ in $\{0,1\}^n$. For any $\sqrt{n} \leq k
\leq n$ and any fixed set $V'$ of up to $k/22\lg k$ vectors in $V$,
the probability that $v$ is in the span of $V'$ when considering only
the last $k$ coordinates is at most $\Delta^{|V'|}/2^k \leq 2^{(k/22\lg
  k) \cdot \lg \Delta-k} \leq 2^{-k/2}$. Since there are less than
$\sum_{i=1}^{k/22\lg k}\binom{n^2}{i} < (k/22\lg k)\binom{n^2}{k/22\lg
  k} \leq (k/22\lg k)\binom{k^8}{k/22\lg k} \leq (22e k^7\lg n)^{k/22
  \lg k+1} < 2^{k/3}$ such sets in $V$, it follows from a union bound
that with probability at least $1-2^{-k/6}$, $v$ will not be in the
span of any set of up to $k/22\lg k$ vectors in $V$ when considering
only the last $k$ coordinates. Finally, by a union bound over all
$\sqrt{n} \leq k \leq n$, it follows that there must exists a vector
that we can add to $V$, which completes the proof of
Lemma~\ref{thm:hardqueries}.

The remainder of this section is dedicated to proving a lower bound of
$t_q = \Omega((\lg n/\lg(wt_u))^2)$ for any data structure solving
this hard range counting problem. Here $t_u$ is the worst case update
time, $t_q$ is the average expected query time, $w$ the cell size and
$n$ the number of points. The proof carries most of the ideas used in
the proof of the main result.

\subsection{The Lower Bound Proof}
The first step is to design a hard distribution over updates, followed
by one uniform random query. We then lower bound the expected cost
(over the distribution) of answering the query for any
\emph{deterministic} data structure with worst case update time
$t_u$. By fixing the random coins (Yao's
principle~\cite{Yao:principle}), this translates into a lower bound on
the expected average query time of a possibly randomized data
structure.

\paragraph{Hard Distribution.}
The hard distribution is extremely simple: For $i=0,\dots,n-1$ (in
this order), we simply set the weight of point $p_i$ to a uniform
random integer $\rdelta_i \in [\Delta]$. Following these updates, we
ask a uniform random query $\vs \in V$.

We think of the updates as divided into \emph{epochs} of exponentially
decreasing size. More specifically, we define epoch $1$ as consisting
of the last $\beta$ updates (the updates that set the weights of
points $p_{n-\beta},\dots,p_{n-1}$), where $\beta \geq 2$ is a
parameter to be fixed later. For $2 \leq i < \lg_\beta n$, epoch $i$
consists of the $\beta^i$ updates that precede epoch $i-1$. Finally,
we let epoch $\lg_\beta n$ consists of the $n-\sum_{i=1}^{\lg_\beta
  n-1} \beta^i$ first updates.

For notational convenience, we let $\U_i$ denote the random variable
giving the sequence of updates performed in epoch $i$ and
$\U=\U_{\lg_\beta},\dots,\U_1$ the random variable giving the updates
of all epochs. Also, we let $\vs$ denote the random variable giving the
uniform random query in $V$.

\paragraph{A Chronogram Approach.}
Having defined the hard distribution over updates and queries, we now
give a high-level proof of the lower bound. Assume a deterministic
data structure solution exists with worst case update time $t_u$. From
this data structure and a sequence of updates $\U$, we define $S(\U)$
to be the set of cells stored in the data structure after executing
the updates $\U$. Now associate each cell in $S(\U)$ to the last epoch
in which its contents were updated, and let $S_i(\U)$ denote the
subset of $S(\U)$ associated to epoch $i$ for $i=1,\dots,\lg_\beta
n$. Also let $t_i(\U,v_j)$ denote the number of cells in $S_i(\U)$
probed by the query algorithm of the data structure when answering the
query $v_j \in V$ after the sequence of updates $\U$. Finally, let
$t_i(\U)$ denote the average cost of answering a query $v_j \in \V$
after the sequence of updates $\U$, i.e. let $t_i(\U) = \sum_{v_j \in
  V}t_i(\U,v_j)/n^2$. Then the following holds:

\begin{lemma}
\label{lem:exprobes}
If $\beta = (w t_u)^2$, then $\E[t_i(\U,\vs)] = \Omega(\lg_\beta n)$
for all $\tfrac{2}{3} \lg_\beta n \leq i < \lg_\beta n$.
\end{lemma}

Before giving the proof of Lemma~\ref{lem:exprobes}, we show that it
implies our lower bound: Let $\beta$ be as in
Lemma~\ref{lem:exprobes}. Since the cell sets $S_{\lg_\beta
  n}(\U),\dots,S_1(\U)$ are disjoint, we get that the number of cells
probed when answering the query $\vs$ is $\sum_i t_i(\U,\vs)$. It now
follows immediately from linearity of expectation that the expected
number of cells probed when answering $\vs$ is $\Omega(\lg_\beta n
\cdot \lg_\beta n)=\Omega((\lg n/\lg(wt_u))^2)$, which completes the
proof.

The hard part thus lies in proving Lemma~\ref{lem:exprobes}, i.e. in
showing that the random query must probe many cells associated to each
of the epochs $i=\tfrac{2}{3} \lg_\beta n,\dots,\lg_\beta n-1$.

\paragraph{Bounding the Probes to Epoch $i$.}
As also pointed out in Section~\ref{sec:tech}, we prove
Lemma~\ref{lem:exprobes} using an encoding argument. Assume for
contradiction that there exists a data structure solution such that
under the hard distribution, with $\beta=(wt_u)^2$, there exists an
epoch $\tfrac{2}{3} \lg_\beta n \leq \is < \lg_\beta n$, such that the
claimed data structure satisfies $\E[t_\is(\U,\vs)] = o(\lg_\beta n)$.

First observe that $\U_\is$ is independent of $\U_{\lg_\beta n} \cdots
\U_{\is+1}$, i.e. $H(\U_\is \mid \U_{\lg_\beta n} \cdots
\U_{\is+1})=H(\U_\is)$, where $H(\cdot)$ denotes binary Shannon
entropy. Furthermore, we have $H(\U_\is)=\beta^\is \lg \Delta$, since
the updates of epoch $\is$ consists of changing the weight of
$\beta^\is$ fixed points, each to a uniform random weight amongst the
integers $[\Delta]$. Our goal is to show that, conditioned on
$\U_{\lg_\beta n} \cdots \U_{\is+1}$, we can use the claimed data
structure solution to encode $\U_\is$ in less than $H(\U_\is)$ bits in
expectation, i.e. a contradiction. We view this encoding step as a
game between an encoder and a decoder. The encoder receives as input
the sequence of updates $\U=\U_{\lg_\beta n},\dots,\U_1$. The encoder
now examines these updates and from them sends a message to the
decoder (an encoding). The decoder sees this message, as well as
$\U_{\lg_\beta n}, \dots, \U_{\is+1}$ (we conditioned on these
variables), and must from this uniquely recover $\U_\is$. If we can
design a procedure for constructing and decoding the encoder's
message, such that the expected size of the message is less than
$H(\U_\is)=\beta^\is \lg \Delta$ bits, then we have reached a
contradiction.

Before presenting the encoding and decoding procedures, we show
exactly what breaks down if the claimed data structure probes too few
cells from epoch $\is$:

\begin{lemma}
\label{lem:exstatic}
Let $\tfrac{2}{3}\lg_\beta n \leq i < \lg_\beta n$ be an epoch. If
$t_i(\U)=o(\lg_\beta n)$, then there exists a subset of cells $C_i(\U)
\subseteq S_i(\U)$ and a set of queries $Q(\U) \subseteq V$ such that:
\begin{enumerate}
\item $|C_i(\U)| = O(\beta^{i-1})$.
\item $|Q(\U)| = \Omega(n)$.
\item The query algorithm of the data structure solution probes no
  cells in $S_i(\U) \setminus C_i(\U)$ when answering a query $v_j \in
  Q(\U)$ after the sequence of updates $\U$.
\end{enumerate}
\end{lemma}

\begin{proof}
Pick a uniform random set $C_i'(\U)$ of $\beta^{i-1}$ cells in
$S_i(\U)$. Now consider the set $Q'(\U)$ of those queries $v_j$ in
$V$ for which $t_i(\U,v_j) \leq \tfrac{1}{4}\lg_\beta n$. Since
$t_i(\U)=o(\lg_\beta n)$ is the average of $t_i(\U,v_j)$ over all
queries $v_j$, it follows from Markov's inequality that $|Q'(\U)| =
\Omega(|V|) = \Omega(n^2)$. Let $v_j$ be a query in $Q'(\U)$. The
probability that all cells probed from $S_i(\U)$ when answering $v_j$
are also in $C_i'(\U)$ is precisely 
\begin{eqnarray*}
\frac{\binom{|S_i(\U)|-t_i(\U,v_j)}{\beta^{i-1}-t_i(\U,v_j)}}{\binom{|S_i(\U)|}{\beta^{i-1}}}
&=& \frac{\beta^{i-1}(\beta^{i-1}-1) \cdots
  (\beta^{i-1}-t_i(\U,v_j)+1)}{|S_i(\U)|(|S_i(\U)|-1) \cdots (|S_i(\U)|-t_i(\U,v_j)+1)} \\
&\geq& \frac{\beta^{i-1}(\beta^{i-1}-1) \cdots
  (\beta^{i-1}-\tfrac{1}{4} \lg_\beta n+1)}{\beta^it_u(\beta^it_u-1) \cdots (\beta^it_u-\tfrac{1}{4}\lg_\beta n+1)} \\
&\geq& \left(\frac{\beta^{i-1}-\tfrac{1}{4} \lg_\beta n}{\beta^it_u}\right)^{\tfrac{1}{4} \lg_\beta n} \\
&\geq& \left(\frac{1}{2\beta t_u}\right)^{\tfrac{1}{4} \lg_\beta n} \\
&\geq& 2^{-\tfrac{1}{2} \lg_\beta n \lg \beta} \\
&=& n^{-\tfrac{1}{2}}.
\end{eqnarray*}
It follows that the expected number of queries in $Q'(\U)$ that probe
only cells in $C_i'(\U)$ is $n^{3/2}$ and hence there must exist a set
satisfying the properties in the lemma.
\end{proof}

The contradiction that this lemma intuitively gives us, is that the
queries in $Q(\U)$ reveal more information about $\U_i$ than the bits
in $C_i(\U)$ can describe (recall the independence properties of the
queries in Lemma~\ref{thm:hardqueries}). We now present the encoding
and decoding procedures:

\paragraph{Encoding.}
Let $\tfrac{2}{3} \lg_\beta n \leq \is < \lg_\beta n$ be the epoch for
which $\E[t_\is(\U,\vs)] = o(\lg_\beta n)$. We construct the message
of the encoder by the following procedure:
\begin{enumerate}
\item First the encoder executes the sequence of updates $\U$ on the
  claimed data structure, and from this obtains the sets $S_{\lg_\beta
    n}(\U),\dots,S_1(\U)$. He then simulates the query algorithm on
  the data structure for every query $v_j \in V$. From this, the
  encoder computes $t_\is(\U)$ (just the average number of cells in
  $S_\is(\U)$ that are probed).
\item If $t_\is(\U) > 2\E[t_\is(\U,\vs)]$, then the encoder writes a
  $1$-bit, followed by $\lceil \beta^\is \lg \Delta\rceil =
  H(\U_\is)+O(1)$ bits, simply specifying each weight assigned to a
  point during the updates $\U_\is$ (this can be done in the claimed
  amount of bits by interpreting the weights as one big integer in
  $[\Delta^{\beta^\is}]$).  This is the complete message send to the
  decoder when $t_\is(\U) > 2\E[t_\is(\U,\vs)]$.
\item If $t_\is(\U) \leq 2\E[t_\is(\U,\vs)]$, then the encoder first
  writes a $0$-bit. Now since $t_\is(\U) \leq 2\E[t_\is(\U,\vs)] =
  o(\lg_\beta n)$, we get from Lemma~\ref{lem:exstatic} that there
  must exist a set of cells $C_\is(\U) \subseteq S_\is(\U)$ and a set
  of queries $Q(\U) \subseteq V$ satisfying the properties in
  Lemma~\ref{lem:exstatic}. The encoder finds such sets $C_\is(\U)$
  and $Q(\U)$ simply by trying all possible sets in some arbitrary but
  fixed order (given two candidate sets $C_\is'(\U)$ and $Q'(\U)$ it
  is straight forward to verify whether they satisfy the properties of
  Lemma~\ref{lem:exstatic}). The encoder now writes down the set
  $C_\is(\U)$, including addresses and contents, for a total of at
  most $O(w)+2|C_\is(\U)|w$ bits (the $O(w)$ bits specifies
  $|C_\is(\U)|$). Following that, he picks
  $\beta^\is/22\lg(\beta^\is)$ arbitrary vectors in $Q(\U)$ (denote
  this set $V'$) and writes down their indices in $V$. This costs
  another $(\beta^\is/22\lg(\beta^\is))\lg(|V|) \leq
  (\beta^\is/22\lg(n^{2/3}))\lg(n^2) = \tfrac{3}{22}\beta^\is$ bits.

\item The encoder now constructs a set $\VecSet$ of vectors in
  $[\Delta]^{k_\is}$, where $k_\is=\sum_{i=1}^\is \beta^\is$ is the
  total size of all epochs $j \leq \is$. He initialized this set by
  first constructing the set of vectors $V'_{k_\is}$ consisting of the
  vectors in $V'$ restricted onto the last $k_\is$ coordinates. He
  then sets $\VecSet = V'_{k_\is}$ and continues by iterating through
  all vectors in $[\Delta]^{k_\is}$, in some arbitrary but fixed
  order, and for each such vector $x=(x_0,\dots,x_{k_\is-1})$,
  checks whether $x$ is in $\vecspan(\VecSet)$. If not, the encoder
  adds $x$ to $\VecSet$. This process continues until
  $\dim(\vecspan(\VecSet))=k_\is$. Now let
  $\uv=(\uv_0,\dots,\uv_{k_\is-1})$ be the $k_\is$-dimensional vector
  with one coordinate for each weight assigned during the last $k_\is$
  updates, i.e. the $i$'th coordinate, $\uv_i$, is given by
  $\uv_i=\rdelta_{n-k_\is+i}$ for $i=0,\dots,k_\is-1$. The encoder now
  computes and writes down $(\langle \uv,x\rangle \mod \Delta)$ for
  each $x$ that was added to $\VecSet$. Here $\langle \uv,x \rangle =
  \sum_i \uv_ix_i$ denotes the standard inner product. Since
  $\dim(\vecspan(V'_{k_\is}))=|V'_{k_\is}| =
  \beta^\is/22\lg(\beta^\is)$ (by Lemma~\ref{thm:hardqueries}), this
  adds a total of $\lceil (k_\is-\beta^\is/22\lg(\beta^\is)) \lg
  \Delta \rceil$ bits to the message.

\item Finally, the encoder writes down all of the cell sets
  $S_{\is-1}(\U),\dots,S_1(\U)$, including addresses and
  contents. This takes at most $\sum_{j=1}^{\is-1}(2|S_j(\U)|w+O(w))$
  bits. When this is done, the encoder sends the constructed message
  to the decoder.
\end{enumerate}

Before analyzing the size of the encoding, we show how the decoder
recovers $\U_\is$ from the above message.

\paragraph{Decoding.}
In the following, we describe the decoding procedure. The decoder
receives as input the updates $\U_{\lg_\beta n},\dots,\U_{\is+1}$ (the
encoding is conditioned on these variables) and the message from the
encoder. The decoder now recovers $\U_\is$ by the following procedure:

\begin{enumerate}
\item The decoder examines the first bit of the message. If this bit
  is $1$, then the decoder immediately recovers $\U_\is$ from the
  encoding (step 2 in the encoding procedure). If not, the decoder instead
  executes the updates $\U_{\lg_\beta n} \cdots \U_{\is+1}$ on
  the claimed data structure solution and obtains the cells sets
  $S_{\lg_\beta n}^{\is+1}(\U),\dots,S_{\is+1}^{\is+1}(\U)$ where
  $S_j^{\is+1}(\U)$ contains the cells
  that were last updated during epoch $j$ when executing updates $\U_{\lg_\beta n}, \dots,
  \U_{\is+1}$ (and not the entire sequence of updates
  $\U_{\lg_\beta n}, \dots, \U_1$).

\item The decoder now recovers $V',C_\is(\U)$ and
  $S_{\is-1}(\U),\dots,S_1(\U)$ from the encoding. For each query $v_j
  \in V'$, the decoder then computes the answer to $v_j$ as if all
  updates $\U_{\lg_\beta n},\dots, \U_1$ had been performed. The
  decoder accomplishes this by simulating the query algorithm on each
  $v_j$, and for each cell requested, the decoder recovers the
  contents of that cell as it would have been if all updates
  $\U_{\lg_\beta n},\dots, \U_1$ had been performed. This is done in
  the following way: When the query algorithm requests a cell $c$, the
  decoder first determines whether $c$ is in one of the sets
  $S_{\is-1}(\U),\dots,S_1(\U)$. If so, the correct contents of $c$
  (the contents after the updates $\U=\U_{\lg_\beta n}, \dots, \U_1$)
  is directly recovered. If $c$ is not amongst these cells, the
  decoder checks whether $c$ is in $C_\is(\U)$. If so, he has again
  recovered the contents. Finally, if $c$ is not in $C_\is(\U)$, then
  from point 3 of Lemma~\ref{lem:exstatic}, we get that $c$ is not in
  $S_\is(\U)$. Since $c$ is not in any of $S_\is(\U),\dots,S_1(\U)$,
  this means that the contents of $c$ has not changed during the
  updates $\U_\is, \dots, \U_1$, and thus the decoder finally recovers
  the contents of $c$ from $S_{\lg_\beta
    n}^{\is+1}(\U),\dots,S_{\is+1}^{\is+1}(\U)$. The decoder can
  therefore recover the answer to each query $v_j$ in $V'$ if it had
  been executed after the sequence of updates $\U_{\lg_\beta n},
  \dots, \U_1$, i.e. for all $v_j \in V'$, he knows $\sum_{i=0}^{n-1}
  v_{j,i} \cdot \rdelta_i$, where $v_{j,i}$ denotes the $i$'th
  coordinate of $v_j$.

\item The next decoding step consists of computing for each query
  $v_j$ in $V'$, the value $\sum_{i=0}^{k_\is-1} v_{j,n-k_\is+i} \cdot
  \uv_i$. Note that this value is precisely the answer to the query
  $v_j$ if all weights assigned during epochs $\lg_\beta
  n,\dots,\is+1$ were set to $0$. Since we conditioned on
  $\U_{\lg_\beta n} \cdots \U_{\is+1}$ the decoder computes this value
  simply by subtracting $\sum_{i=0}^{n-k_\is-1} v_{j,i} \cdot \rdelta_i$
  from $\sum_{i=0}^{n-1} v_{j,i} \cdot \rdelta_i$ (the first sum can
  be computed since $\rdelta_i$ is given from $\U_{\lg_\beta n} \cdots
  \U_{\is+1}$ when $i \leq n-k_\is - 1$).

\item Now from the query set $V'$, the decoder construct the set of
  vectors $\VecSet=V'_{k_\is}$, and then iterates through all vectors
  in $[\Delta]^{k_\is}$, in the same fixed order as the encoder. For
  each such vector $x$, the decoder again verifies whether $x$ is in
  $\vecspan(\VecSet)$, and if not, adds $x$ to $\VecSet$ and recovers
  $(\langle x,\uv \rangle \mod \Delta)$ from the encoding. The decoder
  now constructs the $k_\is \times k_\is$ matrix $\Matrix$, having the
  vectors in $\VecSet$ as rows. Similarly, he construct the vector
  $\z$ having one coordinate for each row of $\Matrix$. The coordinate
  of $\z$ corresponding to a row vector $x$, has the value $(\langle
  x, \uv \rangle \mod \Delta)$. Note that this value is already known
  to the decoder, regardless of whether $x$ was obtained by
  restricting a vector $v_j$ in $V'$ onto the last $k_\is$ coordinates
  (simply taking modulo $\Delta$ on the value $\sum_{i=0}^{k_\is-1}
  v_{j,n-k_\is+i} \cdot \uv_i$ computed for the vector $v_j$ in $V'$
  from which $x$ was obtained), or was added later. Since $\Matrix$
  has full rank, and since the set $[\Delta]$ endowed with integer
  addition and multiplication modulo $\Delta$ is a finite field, it
  follows that the linear system of equations $\Matrix \otimes \y = \z$ has a
  unique solution $\y \in [\Delta]^{k_\is}$ (here $\otimes$ denotes
  matrix-vector multiplication modulo $\Delta$). But $\uv \in
  [\Delta]^{k_\is}$ and $\Matrix \otimes \uv = \z$, thus the decoder
  now solves the linear system of equations $\Matrix \otimes \y = \z$
  and uniquely recovers $\uv$, and therefore also $\U_\is \cdots
  \U_1$. This completes the decoding procedure.
\end{enumerate}

\paragraph{Analysis.}
We now analyse the expected size of the encoding of
$\U_\is$. We first analyse the size of the encoding when
$t_\is(\U) \leq 2\E[t_\is(\U,\vs)]$. In this case, the encoder sends a
message of
\begin{eqnarray*}
2|C_\is(\U)|w+\tfrac{3}{22}\beta^\is+(k_\is-\beta^\is/22\lg(\beta^\is)) \lg
  \Delta
+\sum_{j=1}^{\is-1}2|S_j(\U)|w+O(w \lg_\beta n)
\end{eqnarray*}
bits. Since $H(\U_\is) = \beta^\is \lg \Delta = k_\is \lg \Delta -
\sum_{j=1}^{\is-1}\beta^j \lg \Delta$,
$|C_\is(\U)|w=O(\beta^{\is-1}w)$, $\lg \Delta = O(w)$ and $|S_j(\U)|
\leq \beta^j t_u$, the above is upper bounded by
\[
H(\U_\is)+\tfrac{3}{22}\beta^\is-(\beta^\is/22\lg(\beta^\is))\lg\Delta
+O\left(\sum_{j=1}^{\is-1}\beta^jwt_u\right). 
\]
Since $\beta \geq 2$, we also have $O\left(\sum_{j=1}^{\is-1} \beta^j
wt_u\right) = O\left(\beta^{\is-1} wt_u
\right)=o(\beta^\is)$. Finally, we have $\lg \Delta \geq \lg n^4-1 \geq 
4\lg(\beta^\is)-1$. Therefore, the above is again upper bounded by
\[
H(\U_\is)+\tfrac{3}{22}\beta^\is-\tfrac{4}{22}\beta^\is+o(\beta^\is)=H(\U_\is)-\Omega(\beta^{\is}).
\]
This part thus contributes at most
\[
\Pr[t_\is(\U) \leq 2 \E[t_\is(\U,\q)]] \cdot
(H(\U_\is)-\Omega(\beta^{\is}))
\]
bits to the expected size of the encoding. The case where
$t_\is(\U)>2\E[t_\is(\U,\vs)]$ similarly contributes $ \Pr[t_\is(\U) >
  2 \E[t_\is(\U,\vs)]] \cdot (H(\U_\is)+O(1)) $ bits to the expected
size of the encoding. Now since $\vs$ is uniform, we have
$\E[t_\is(\U)]=\E[t_\is(\U,\vs)]$, we therefore get from Markov's
inequality that $\Pr[t_\is(\U) > 2 \E[t_\is(\U,\vs)]]<\tfrac{1}{2}$.
Therefore the expected size of the encoding is upper bounded by $
O(1)+\tfrac{1}{2}H(\U_\is)+\tfrac{1}{2}(H(\U_\is)-\Omega(\beta^{\is}))
< H(\U_\is) $ bits, finally leading to the contradiction and
completing the proof of Lemma~\ref{lem:exprobes}.

\section{Weighted Orthogonal Range Counting}
\label{sec:dyn}
In this section we prove our main result, which we have formulated in
the following theorem:
\begin{theorem}
\label{thm:main}
Any data structure for dynamic weighted orthogonal range counting in
the cell probe model, must satisfy $t_q = \Omega((\lg n/\lg(w
t_u))^2)$. Here $t_q$ is the expected average query time and $t_u$ the
worst case update time. This lower bound holds when the weights of the
inserted points are $\Theta(\lg n)$-bit integers.
\end{theorem}
As in Section~\ref{sec:ex}, we prove Theorem~\ref{thm:main} by
devising a hard distribution over updates, followed by one uniform
random query. We then lower bound the expected cost (over the
distribution) of answering the query for any \emph{deterministic} data
structure with worst case update time $t_u$. In the proof we assume
the weights are $4 \lg n$-bit integers and note that the lower bound
applies to any $\eps \lg n$-bit weights, where $\eps > 0$ is an
arbitrarily small constant, simply because a data structure for $\eps
\lg n$-bit integer weights can be used to solve the problem for any
$O(\lg n)$-bit integer weights with a constant factor overhead by
dividing the bits of the weights into $\lceil \delta/(\eps \lg n)
\rceil = O(1)$ chunks and maintaining a data structure for each
chunk. We begin the proof by presenting the hard distribution over
updates and queries.

\paragraph{Hard Distribution.}
Again, updates arrive in epochs of exponentially decreasing
size. For $i=1,\dots,\lg_\beta n$ we define epoch $i$ as a sequence of
$\beta^i$ updates, for a parameter $\beta>1$ to be fixed later. The
epochs occur in time from biggest to smallest epoch, and at the end of
epoch $1$ we execute a uniform random query in $[n] \times [n]$.

What remains is to specify which updates are performed in each epoch
$i$. The updates of epoch $i$ are chosen to mimic the hard input
distribution for static orthogonal range counting on a set of
$\beta^i$ points.  We first define the following point set known as
the Fibonacci lattice:

\begin{definition}[\cite{matousek:discrepency}]
The Fibonacci lattice $F_m$ is the set of $m$ two-dimensional points
defined by $F_m=\{(i,if_{k-1} \mod m) \mid i=0,\dots,m-1\}$, where
$m=f_k$ is the $k$'th Fibonacci number.
\end{definition}

The $\beta^i$ updates of epoch $i$ now consists of inserting each
point of the Fibonacci lattice $F_{\beta^i}$, but scaled to fit the
input region $[n] \times [n]$, i.e. the $j$'th update of epoch $i$
inserts the point with coordinates $(n/\beta^i \cdot j,n/\beta^i \cdot
(jf_{k_i-1} \mod \beta^i))$, for $j=0,\dots,\beta^i$. The weight of
each inserted point is a uniform random integer amongst
$[\Delta]$, where $\Delta$ is the largest prime number
smaller than $2^{4 \lg n}=n^4$. This concludes the description of the hard
distribution.

The Fibonacci lattice has the desirable property that it is very
uniform. This plays an important role in the lower bound proof, and we
have formulated this property in the following lemma:

\begin{lemma}[\cite{fibonacci}]
\label{lem:fibarea}
For the Fibonacci lattice $F_{\beta^i}$, where the coordinates of each
point have been multiplied by $n/\beta^i$, and for $\alpha>0$, any
axis-aligned rectangle in $[0,n-n/\beta^i] \times [0,n-n/\beta^i]$
with area $\alpha n^2/\beta^i$ contains between $\lfloor \alpha / a_1
\rfloor$ and $\lceil \alpha/a_2 \rceil$ points, where $a_1 \approx
1.9$ and $a_2 \approx 0.45$.
\end{lemma}

Note that we assume each $\beta^i$ to be a Fibonacci number (denoted
$f_{k_i}$), and that each $\beta^i$ divides $n$. These assumptions can
easily be removed by fiddling with the constants, but this would only
clutter the exposition.

For the remainder of the paper, we let $\U_i$ denote the random
variable giving the sequence of updates in epoch $i$, and we let
$\U=\U_{\lg_\beta n}\cdots\U_1$ denote the random variable giving all
updates of all $\lg_\beta n$ epochs. Finally, we let $\q$ be the
random variable giving the query.

\paragraph{A Chronogram.}
Having defined the hard distribution over updates and queries, we
proceed as in Section~\ref{sec:ex}. Assume a deterministic data
structure solution exists with worst case update time $t_u$. From this
data structure and a sequence of updates $\U=\U_{\lg_\beta
  n},\dots,\U_1$, we define $S(\U)$ to be the set of cells stored in
the data structure after executing the updates $\U$. Associate each
cell in $S(\U)$ to the last epoch in which its contents were updated,
and let $S_i(\U)$ denote the subset of $S(\U)$ associated to epoch $i$
for $i=1,\dots,\lg_\beta n$. Also let $t_i(\U,q)$ denote the number of
cells in $S_i(\U)$ probed by the query algorithm of the data structure
when answering the query $q \in [n] \times [n]$ after the sequence of
updates $\U$. Finally, let $t_i(\U)$ denote the average cost of
answering a query $q \in [n] \times [n]$ after the sequence of updates
$\U$, i.e. let $t_i(\U) = \sum_{q \in [n] \times [n]}t_i(\U,q)/n^2$. Then
the following holds:

\begin{lemma}
\label{lem:probes}
If $\beta = (w t_u)^9$, then $\E[t_i(\U,\q)] =
\Omega(\lg_\beta n)$ for all $i \geq \tfrac{15}{16} \lg_\beta n$.
\end{lemma}

The lemma immediately implies Theorem~\ref{thm:main} since the cell
sets $S_{\lg_\beta n}(\U),\dots,S_1(\U)$ are disjoint and the number
of cells probed when answering the query $\q$ is $\sum_i
t_i(\U,\q)$. We prove Lemma~\ref{lem:probes} in the following section.

\subsection{Bounding the Probes to Epoch $i$}
The proof of Lemma~\ref{lem:probes} is again based on an encoding
argument. The framework is identical to Section~\ref{sec:ex}, but
arguing that a ``good'' set of queries to simulate exists is
significantly more difficult.

Assume for contradiction that there exists a data structure solution
such that under the hard distribution, with $\beta=(wt_u)^9$, there
exists an epoch $\is \geq \tfrac{15}{16} \lg_\beta n$, such that the
claimed data structure satisfies $\E[t_\is(\U,\q)] = o(\lg_\beta n)$.

First observe that $\U_\is$ is independent of $\U_{\lg_\beta n} \cdots
\U_{\is+1}$, i.e. $H(\U_\is \mid \U_{\lg_\beta n} \cdots
\U_{\is+1})=H(\U_\is)$. Furthermore, we have $H(\U_\is)=\beta^\is \lg
\Delta$, since the updates of epoch $\is$ consists of inserting
$\beta^\is$ fixed points, each with a uniform random weight amongst
the integers $[\Delta]$. Our goal is to show that, conditioned on
$\U_{\lg_\beta n} \cdots \U_{\is+1}$, we can use the claimed data
structure solution to encode $\U_\is$ in less than $H(\U_\is)$ bits in
expectation, which provides the contradiction.

Before presenting the encoding and decoding procedures, we show what
happens if a data structure probes too few cells from epoch $\is$. For
this, we first introduce some terminology. For a query point $q=(x,y)
\in [n] \times [n]$, we define for each epoch $i=1,\dots,\lg_\beta n$
the \emph{incidence vector} $\inc{i}(q)$, as a $\{0,1\}$-vector in
$[\Delta]^{\beta^i}$. The $j$'th coordinate of $\inc{i}(q)$ is $1$ if
the $j$'th point inserted in epoch $i$ is dominated by $q$, and $0$
otherwise. More formally, for a query $q=(x,y)$, the $j$'th coordinate
$\inc{i}(q)_j$ is given by:
\[
\inc{i}(q)_j = \left \{ \begin{array}{l l}
  1 & \quad \text{if }jn/\beta^i \leq x \wedge (jf_{k_i-1}\textrm{ mod } \beta^i)n/\beta^i\leq y\\
  0 & \quad \text{otherwise} 
  \end{array} \right.
\]

Similarly, we define for a sequence of updates $\U_i$, the
$\beta^i$-dimensional vector $\uv_i$ for which the $j$'th coordinate
equals the weight assigned to the $j$'th inserted point in $\U_i$. We
note that $\U_i$ and $\uv_i$ uniquely specify each other, since
$\U_i$ always inserts the same fixed points, only the weights vary.

Finally observe that the answer to a query $q$ after a sequence of
updates $\U_{\lg_\beta n},\dots, \U_1$ is $\sum_{i=1}^{\lg_\beta n}
\langle \inc{i}(q), \uv_i\rangle$.  With these definitions, we now
present the main result forcing a data structure to probe many cells
from each epoch:
\begin{lemma}
\label{lem:static}
Let $i \geq \tfrac{15}{16}\lg_\beta n$ be an epoch. If
$t_i(\U)=o(\lg_\beta n)$, then there exists a subset of cells $C_i(\U)
\subseteq S_i(\U)$ and a set of query points $Q(\U) \subseteq [n] \times
          [n]$ such that:
\begin{enumerate}
\item $|C_i(\U)| = O(\beta^{i-1}w)$.
\item $|Q(\U)| = \Omega(\beta^{i-3/4})$.
\item The set of incidence vectors $\inc{i}(Q(\U)) = \{ \inc{i}(q)
  \mid q \in Q(\U)\}$ is a linearly independent set of vectors in
  $[\Delta]^{\beta^i}$.
\item The query algorithm of the data structure solution probes no
  cells in $S_i(\U) \setminus C_i(\U)$ when answering a query $q \in
  Q(\U)$ after the sequence of updates $\U$.
\end{enumerate}
\end{lemma}

Comparing to Lemma~\ref{lem:exstatic}, it is not surprising that this
lemma gives the lower bound. We note that Lemma~\ref{lem:static}
essentially is a generalization of the results proved in the static
range counting papers~\cite{patrascu07sum2D,larsen:median}, simply
phrased in terms of cell subsets answering many queries instead of
communication complexity. Since the proof contains only few new ideas,
we have deferred it to Section~\ref{sec:static} and instead move on to
the encoding and decoding procedures.

\paragraph{Encoding.}
Let $\is \geq \tfrac{15}{16} \lg_\beta n$ be the epoch for which
$\E[t_\is(\U,\q)] = o(\lg_\beta n)$. The encoding procedure follows
that in Section~\ref{sec:ex} uneventfully:
\begin{enumerate}
\item First the encoder executes the sequence of updates $\U$ on the
  claimed data structure, and from this obtains the sets $S_{\lg_\beta
    n}(\U),\dots,S_1(\U)$. He then simulates the query algorithm on the
  data structure for every query $q \in [n] \times [n]$. From this,
  the encoder computes $t_\is(\U)$.
\item If $t_\is(\U) > 2\E[t_\is(\U,\q)]$, then the encoder writes a
  $1$-bit, followed by $\lceil \beta^\is \lg \Delta\rceil =
  H(\U_\is)+O(1)$ bits, simply specifying each weight assigned to a
  point in $\U_\is$.  This is the complete message send to the decoder
  when $t_\is(\U) > 2\E[t_\is(\U,\q)]$.
\item If $t_\is(\U) \leq 2\E[t_\is(\U,\q)]$, then the encoder first
  writes a $0$-bit. Now since $t_\is(\U) \leq 2\E[t_\is(\U,\q)] =
  o(\lg_\beta n)$, we get from Lemma~\ref{lem:static} that there must
  exist a set of cells $C_\is(\U) \subseteq S_\is(\U)$ and a set of
  queries $Q(\U) \subseteq [n] \times [n]$ satisfying $1$-$4$ in
  Lemma~\ref{lem:static}. The encoder finds such sets $C_\is(\U)$ and
  $Q(\U)$ simply by trying all possible sets in some arbitrary but
  fixed order. The encoder now writes down these two sets,
  including addresses and contents of the cells in $C_\is(\U)$, for a
  total of at most $O(w)+2|C_\is(\U)|w+\lg \binom{n^2}{|Q(\U)|}$ bits (the
  $O(w)$ bits specifies $|C_\is(\U)|$ and $|Q(\U)|$).

\item The encoder now constructs a set $\VecSet$, such that
  $\VecSet=\inc{\is}(Q(\U)) = \{ \inc{\is}{q} \mid q \in Q(\U) \}$
  initially. Then he iterates through all vectors in
  $[\Delta]^{\beta^\is}$, in some arbitrary but fixed order, and for
  each such vector $x$, checks whether $x$ is in
  $\vecspan(\VecSet)$. If not, the encoder adds $x$ to $\VecSet$. This
  process continues until $\dim(\vecspan(\VecSet))=\beta^\is$, at
  which point the encoder computes and writes down $(\langle x,\uv_\is
  \rangle \mod \Delta)$ for each $x$ that was added to
  $\VecSet$. Since $\dim(\vecspan(\inc{\is}(Q(\U))))=|Q(\U)|$ (by
  point 3 in Lemma~\ref{lem:static}), this adds a total of $\lceil
  (\beta^\is-|Q(\U)|) \lg \Delta \rceil$ bits to the message.

\item Finally, the encoder writes down all of the cell sets
  $S_{\is-1}(\U),\dots,S_1(\U)$, including addresses and contents,
  plus all of the vectors $\uv_{\is-1},\dots, \uv_1$. This takes at
  most $\sum_{j=1}^{\is-1}(2|S_j(\U)|w+\beta^j \lg \Delta+O(w))$
  bits. When this is done, the encoder sends the constructed message
  to the decoder.
\end{enumerate}

Next we present the decoding procedure:

\paragraph{Decoding.} The decoder receives as input the updates $\U_{\lg_\beta n},\dots,\U_{\is+1}$ and the message from the encoder. The decoder now recovers $\U_\is$  by the following procedure:

\begin{enumerate}
\item The decoder examines the first bit of the message. If this bit
  is $1$, then the decoder immediately recovers $\U_\is$ from the
  encoding (step 2 in the encoding procedure). If not, the decoder
  instead executes the updates $\U_{\lg_\beta n} \cdots \U_{\is+1}$ on
  the claimed data structure solution and obtains the cells sets
  $S_{\lg_\beta n}^{\is+1}(\U),\dots,S_{\is+1}^{\is+1}(\U)$ where
  $S_j^{\is+1}(\U)$ contains the cells that were last updated during
  epoch $j$ when executing only the updates $\U_{\lg_\beta n}, \dots,
  \U_{\is+1}$.

\item The decoder now recovers
  $Q(\U),C_\is(\U),S_{\is-1}(\U),\dots,S_1(\U)$ and $\uv_{\is-1},
  \dots, \uv_1$ from the encoding. For each query $q \in Q(\U)$, the
  decoder then computes the answer to $q$ as if all updates
  $\U_{\lg_\beta n},\dots, \U_1$ had been performed. The decoder
  accomplishes this by simulating the query algorithm on $q$, and for
  each cell requested, the decoder recovers the contents of that cell
  as it would have been if all updates $\U_{\lg_\beta n},\dots, \U_1$
  had been performed. This is done as follows: When the query
  algorithm requests a cell $c$, the decoder first determines whether
  $c$ is in one of the sets $S_{\is-1}(\U),\dots,S_1(\U)$. If so, the
  correct contents of $c$ is directly recovered. If $c$ is not amongst
  these cells, the decoder checks whether $c$ is in $C_\is(\U)$. If
  so, the decoder has again recovered the contents. Finally, if $c$ is
  not in $C_\is(\U)$, then from point 4 of Lemma~\ref{lem:static}, we
  get that $c$ is not in $S_\is(\U)$. Since $c$ is not in any of
  $S_\is(\U),\dots,S_1(\U)$, this means that the contents of $c$ has
  not changed during the updates $\U_\is, \dots, \U_1$, and thus the
  decoder finally recovers the contents of $c$ from $S_{\lg_\beta
    n}^{\is+1}(\U),\dots,S_{\is+1}^{\is+1}(\U)$. The decoder can
  therefore recover the answer to each query $q$ in $Q(\U)$ if it had
  been executed after the sequence of updates $\U$, i.e. for all $q
  \in Q(\U)$, he knows $\sum_{i=1}^{\lg_\beta n} \langle \inc{i}(q) ,
  \uv_i \rangle$.

\item The next decoding step consists of computing for each query $q$
  in $Q(\U)$, the value $\langle \inc{\is}(q),\uv_\is \rangle$. For each $q \in
  Q(\U)$, the decoder already knows the value $\sum_{i=1}^{\lg_\beta n}
  \langle \inc{i}(q),\uv_i\rangle$ from the above. From the encoding of
  $\uv_{\is-1}, \dots, \uv_1$, the decoder can compute the value
  $\sum_{i=1}^{\is-1} \langle \inc{i}(q),\uv_i \rangle$ and finally from
  $\U_{\lg_\beta n}, \dots, \U_{\is+1}$ the decoder computes
  $\sum_{i=\is+1}^{\lg_\beta n} \langle \inc{i}(q),\uv_i \rangle$. The
  decoder can now recover the value $\langle \inc{\is}(q),\uv_\is
  \rangle$ simply by observing that $\langle \inc{\is}(q),\uv_\is \rangle=
  \sum_{i=1}^{\lg_\beta n} \langle \inc{i}(q),\uv_i\rangle - \sum_{i
    \neq \is} \langle \inc{i}(q),\uv_i \rangle$.

\item Now from the query set $Q(\U)$, the decoder construct the set of
  vectors $\VecSet=\inc{\is}(Q(\U))$, and then iterates through all
  vectors in $[\Delta]^{\beta^\is}$, in the same fixed order as the
  encoder. For each such vector $x$, the decoder again verifies
  whether $x$ is in $\vecspan(\VecSet)$, and if not, adds $x$ to
  $\VecSet$ and recovers $\langle x,\uv_\is \rangle \mod \Delta$ from
  the encoding. The decoder now constructs the $\beta^\is \times
  \beta^\is$ matrix $\Matrix$, having the vectors in $\VecSet$ as
  rows. Similarly, he construct the vector $\z$ having one coordinate
  for each row of $\Matrix$. The coordinate of $\z$ corresponding to a
  row vector $x$, has the value $\langle x, \uv_\is \rangle \mod
  \Delta$. Since $\Matrix$ has full rank, it follows that the linear system of
  equations $\Matrix \otimes \y = \z$ has a unique solution $\y \in
  [\Delta]^{\beta^\is}$. But $\uv_\is \in
  [\Delta]^{\beta^\is}$ and $\Matrix \otimes \uv_\is = \z$, thus the
  decoder solves the linear system of equations $\Matrix \otimes
  \y = \z$ and uniquely recovers $\uv_\is$, and therefore also
  $\U_\is$. This completes the decoding procedure.
\end{enumerate}

\paragraph{Analysis.}
We now analyse the expected size of the encoding of
$\U_\is$. We first analyse the size of the encoding when
$t_\is(\U) \leq 2\E[t_\is(\U,\q)]$. In this case, the encoder sends a
message of
\[
2|C_\is(\U)|w+\lg\binom{n^2}{|Q(\U)|}+(\beta^\is-|Q(\U)|)\lg \Delta
+O(w \lg_\beta n)+\sum_{j=1}^{\is-1}(2|S_j(\U)|w+\beta^j\lg \Delta)
\]
bits. Since $\beta^\is \lg \Delta = H(\U_\is)$ and
$|C_\is(\U)|w=o(|Q(\U)|)$, the above is upper bounded by
\[
H(\U_\is)-|Q(\U)|\lg(\Delta/n^2)+o(|Q(\U)|)
+\sum_{j=1}^{\is-1}(2|S_j(\U)|w+\beta^j\lg \Delta).
\]
Since $\beta \geq 2$, we also have $\sum_{j=1}^{\is-1} \beta^j \lg
\Delta \leq 2\beta^{\is-1}\lg \Delta=o(|Q(\U)|\lg \Delta)$. Similarly, we
have $|S_j(\U)| \leq \beta^j t_u$, which gives us
$\sum_{j=1}^{\is-1}2|S_j(\U)|w \leq 4\beta^{\is-1}wt_u=o(|Q(\U)|)$. From
standard results on prime numbers, we have that the largest prime
number smaller than $n^4$ is at least $n^3$ for infinitely many $n$,
i.e. we can assume $\lg(\Delta/n^2)=\Omega(\lg \Delta)$. Therefore,
the above is again upper bounded by
\[
H(\U_\is)-\Omega(|Q(\U)|\lg \Delta)=H(\U_\is)-\Omega(\beta^{\is-3/4}\lg \Delta).
\]
This part thus contributes at most
\[
\Pr[t_\is(\U) \leq 2 \E[t_\is(\U,\q)]] \cdot
(H(\U_\is)-\Omega(\beta^{\is-3/4}\lg \Delta))
\]
bits to the expected size of the encoding. The case where
$t_\is(\U)>2\E[t_\is(\U,\q)]$ similarly contributes
$
\Pr[t_\is(\U) > 2 \E[t_\is(\U,\q)]] \cdot (H(\U_\is)+O(1))
$
bits to the expected size of the encoding. Now since $\q$ is uniform, we have
$\E[t_\is(\U)]=\E[t_\is(\U,\q)]$, we therefore get from
Markov's inequality that $\Pr[t_\is(\U) > 2 \E[t_\is(\U,\q)]]<\tfrac{1}{2}$.
Therefore the expected size of the encoding is upper bounded by
$
O(1)+\tfrac{1}{2}H(\U_\is)+\tfrac{1}{2}(H(\U_\is)-\Omega(\beta^{\is-3/4} \lg
\Delta)) < H(\U_\is).
$
This completes the proof of
Lemma~\ref{lem:probes}.

\section{The Static Setup}
\label{sec:static}
Finally, in this section we prove Lemma~\ref{lem:static}, the last
piece in the lower bound proof. As already mentioned, we prove the
lemma by extending on previous ideas for proving lower bounds on
static range counting. We note that we have chosen a more geometric
(and we believe more intuitive) approach to the proof than the
previous papers.

For the remainder of the section, we let $U=U_{\lg_\beta n},\dots,U_1$
be a fixed sequence of updates, where each $U_j$ is a possible outcome
of $\U_j$, and $i \geq \tfrac{15}{16}\lg_\beta n$ an epoch. Furthermore,
we assume that the claimed data structure satisfies
$t_i(U)=o(\lg_\beta n)$, and our task is to show that the claimed cell
set $C_i$ and query set $Q$ exists.

The first step is to find a geometric property of a set of queries
$Q$, such that $\inc{i}(Q)$ is a linearly independent set of
vectors. One property that ensures this, is that the queries in $Q$
are sufficiently \emph{well spread}. To make this more formal, we
introduce the following terminology:

A \emph{grid} $G$ with \emph{width} $\mu \geq 1$ and \emph{height}
$\gamma \geq 1$, is the collection of \emph{grid cells} $[j
  \mu,(j+1)\mu) \times[h\gamma,(h+1)\gamma)$ such that $0 \leq j <
    n/\mu$ and $0\leq h < n/\gamma$. We say that a query point
    $q=(x,y) \in [n] \times [n]$ \emph{hits} a grid cell $[j
      \mu,(j+1)\mu) \times[h\gamma,(h+1)\gamma)$ of $G$, if the point
        $(x,y)$ lies within that grid cell, i.e. if $j\mu \leq x <
        (j+1)\mu$ and $h \gamma \leq y < (h+1)\gamma$. Finally, we
        define the \emph{hitting number} of a set of queries $Q'$ on a
        grid $G$, as the number of distinct grid cells in $G$ that is
        hit by a query in $Q'$.

With this terminology, we have the following lemma:

\begin{lemma}
\label{lem:spread}
Let $Q'$ be a set of queries and $G$ a grid with width $\mu$ and
height $n^2/\beta^i \mu$ for some parameter $n/\beta^i \leq \mu \leq
n$. Let $h$ denote the hitting number of $Q'$ on $G$. Then there is a
subset of queries $Q \subseteq Q'$, such that $|Q| =
\Omega(h-6n/\mu-6\mu\beta^i/n)$ and $\inc{i}(Q)$ is a linearly
independent set of vectors in $[\Delta]^{\beta^i}$.
\end{lemma}

We defer the proof of Lemma~\ref{lem:spread} to
Section~\ref{sec:proofspread}, and instead continue the proof of
Lemma~\ref{lem:static}.

In light of Lemma~\ref{lem:spread}, we set out to find a set of cells
$C_i \subseteq S_i(U)$ and a grid $G$, such that the set of queries
$Q_{C_i}$ that probe no cells in $S_i(U) \setminus C_i$, hit a large
number of grid cells in $G$. For this, first define the grids
$G_2,\dots,G_{2i-2}$ where $G_j$ has width $n/\beta^{i-j/2}$ and
height $n/\beta^{j/2}$.  The existence of $C_i$ is guaranteed by the
following lemma:

\begin{lemma}
\label{lem:hitmany}
Let $i \geq \tfrac{15}{16}\lg_\beta n$ be an epoch and $U_{\lg_\beta
  n},\dots,U_1$ a fixed sequence of updates, where each $U_j$ is a
possible outcome of $\U_j$. Assume furthermore that the claimed data
structure satisfies $t_i(U) = o(\lg_\beta n)$. Then there exists a set
of cells $C_i \subseteq S_i(U)$ and an index $j \in \{2,\dots,2i-2\}$,
such that $|C_i| = O(\beta^{i-1}w)$ and $Q_{C_i}$ has hitting number
$\Omega(\beta^{i-3/4})$ on the grid $G_j$.
\end{lemma}

To not remove focus from the proof of Lemma~\ref{lem:static} we have
moved the proof of this lemma to Section~\ref{sec:proofhitmany}. We
thus move on to show that Lemma~\ref{lem:spread} and
Lemma~\ref{lem:hitmany} implies Lemma~\ref{lem:static}. By assumption
we have $t_i(U)=o(\lg_\beta n)$. Combining this with
Lemma~\ref{lem:hitmany}, we get that there exists a set of cells $C_i
\subseteq S_i(U)$ and an index $j \in \{2,\dots,2i-2\}$, such that
$|C_i| = O(\beta^{i-1}w)$ and the set of queries $Q_{C_i}$ has hitting
number $\Omega(\beta^{i-3/4})$ on the grid $G_j$. Furthermore, we have
that grid $G_j$ is a grid of the form required by
Lemma~\ref{lem:spread}, with $\mu = n/\beta^{i-j/2}$. Thus by
Lemma~\ref{lem:spread} there is a subset $Q \subseteq Q_{C_i}$ such
that $|Q| = \Omega(\beta^{i-3/4}-12\beta^{i-1}) =
\Omega(\beta^{i-3/4})$ and $\inc{i}(Q)$ is a linearly independent set
of vectors in $[\Delta]^{\beta^i}$. This completes the proof of Lemma~\ref{lem:static}.

\subsection{Proof of Lemma~\ref{lem:spread}}
\label{sec:proofspread}
We prove the lemma by giving an explicit construction of the set
$Q$. 

First initialize $Q$ to contain one query point from $Q'$ from each
cell of $G$ that is hit by $Q'$. We will now repeatedly eliminate
queries from $Q$ until the remaining set is linearly independent. We
do this by \emph{crossing out} rows and columns of $G$. By crossing
out a row (column) of $G$, we mean deleting all queries in $Q$ that
hits a cell in that row (column). The procedure for crossing out rows
and columns is as follows:

First cross out the bottom two rows and leftmost two columns. Amongst
the remaining columns, cross out either the even or odd columns,
whichever of the two contains the fewest remaining points in
$Q$. Repeat this once again for the columns, with even and odd
redefined over the remaining columns. Finally, do the same for the
rows. We claim that the remaining set of queries are linearly
independent. To see this, order the remaining queries in increasing
order of column index (leftmost column has lowest index), and
secondarily in increasing order of row index (bottom row has lowest
index). Let $q_1,\dots,q_{|Q|}$ denote the resulting sequence of
queries. For this sequence, it holds that for every query $q_j$, there
exists a coordinate $\inc{i}(q_j)_h$, such that $\inc{i}(q_j)_h=1$,
and at the same time $\inc{i}(q_k)_h=0$ for all $k < j$. Clearly this
implies linear independence. To prove that the remaining vectors have
this property, we must show that for each query $q_j$, there is some
point in the scaled Fibonacci lattice $F_{\beta^i}$ that is dominated
by $q_j$, but not by any of $q_1,\dots,q_{j-1}$: Associate each
remaining query $q_j$ to the two-by-two crossed out grid cells to the
bottom-left of the grid cell hit by $q_j$. These four grid cells have
area $4n^2/\beta^i$ and are contained within the rectangle
$[0,n-n/\beta^i] \times [0,n-n/\beta^i]$, thus from
Lemma~\ref{lem:fibarea} it follows that at least one point of the
scaled Fibonacci lattice $F_{\beta^i}$ is contained therein, and thus
dominated by $q_j$. But all $q_k$, where $k<j$, either hit a grid cell
in a column with index at least three less than that hit by $q_j$ (we
crossed out the two columns preceding that hit by $q_j$), or they hit
a grid cell in the same column as $q_j$ but with a row index that is
at least three lower than that hit by $q_j$ (we crossed out the two
rows preceding that hit by $q_j$). In either case, such a query
cannot dominate the point inside the cells associated to $q_j$.

What remains is to bound the size of $Q$. Initially, we have
$|Q|=h$. The bottom two rows have a total area of $2n^3/\beta^i \mu$,
thus by Lemma~\ref{lem:fibarea} they contain at most $6n/\mu$
points. The leftmost two columns have area $2n\mu$ and thus contain at
most $6\mu\beta^i/n$ points. After crossing out these rows and column
we are therefore left with $|Q| \geq h-6n/\mu-6\mu\beta^i/n$. Finally,
when crossing out even or odd rows we always choose the one
eliminating fewest points, thus the remaining steps at most reduce the
size of $Q$ by a factor $16$. This completes the proof of
Lemma~\ref{lem:spread}.

\subsection{Proof of Lemma~\ref{lem:hitmany}}
\label{sec:proofhitmany}
We prove the lemma using another encoding argument. However, this time
we do not encode an update sequence, but instead we define a
distribution over query sets, such that if Lemma~\ref{lem:hitmany} is
not true, then we can encode such a query set in too few bits.

Let $U=U_{\lg_\beta n},\dots,U_1$ be a fixed sequence of updates,
where each $U_j$ is a possible outcome of $\U_j$. Furthermore, assume
for contradiction that the claimed data structure satisfies both
$t_i(U) = o(\lg_\beta n)$ and for all cell sets $C \subseteq S_i(U)$
of size $|C|=O(\beta^{i-1}w)$ and every index $j \in \{2,\dots,2i-2
\}$, it holds that the hitting number of $Q_C$ on grid $G_j$ is
$o(\beta^{i-3/4})$. Here $Q_C$ denotes the set of all queries $q$ in
$[n]\times[n]$ such that the query algorithm of the claimed data
structure probes no cells in $S_i(U) \setminus C$ when answering $q$
after the sequence of updates $U$. Under these assumptions we will
construct an impossible encoder. As mentioned, we will encode a set of
queries:

\paragraph{Hard Distribution.}
Let $\Q$ denote a random set of queries, constructed by drawing one
uniform random query (with integer coordinates) from each of the
$\beta^{i-1}$ vertical slabs of the form:
$$[hn/\beta^{i-1},(h+1)n/\beta^{i-1}) \times [0,n),$$ where $h \in
    [\beta^{i-1}]$. Our goal is to encode $\Q$ in less than $H(\Q)
    =\beta^{i-1} \lg(n^2/\beta^{i-1})$ bits in expectation. Before
    giving the encoding and decoding procedures, we prove some simple
    properties of $\Q$:

Define a query $q$ in a query set $Q'$ to be \emph{well-separated} if
for all other queries $q' \in Q'$, where $q \neq q'$, $q$ and $q'$ do
not lie within an axis-aligned rectangle of area
$n^2/\beta^{i-1/2}$. Finally, define a query set $Q'$ to be
\emph{well-separated} if at least $\tfrac{1}{2}|Q'|$ queries in $Q'$ are
well-separated. We then have:
\begin{lemma}
\label{lem:qsize}
The query set $\Q$ is well-separated with probability at least $3/4$.
\end{lemma}
\begin{proof}
Let $\q_h$ denote the random query in $\Q$ lying in the $h$'th
vertical slab. The probability that $\q_h$ lies within a distance of at most  $n/\beta^{i-3/4}$ from the $x$-border of 
the $h$'th slab is precisely $(2n/\beta^{i-3/4}) / (n/\beta^{i-1}) =
2/\beta^{1/4}$. If this is not the case, then for another query $\q_k$
in $\Q$, we know that the $x$-coordinates of $\q_h$ and $\q_k$ differ
by at least $(|k-h|-1)n/\beta^{i-1}+n/\beta^{i-3/4}$. This implies
that $\q_h$ and $\q_k$ can only be within an axis-aligned rectangle of
area $n^2/\beta^{i-1/2}$ if their $y$-coordinates differ by at most
$n/((|k-h|-1)\beta^{1/2}+\beta^{1/4})$. This happens with probability
at most $2/((|k-h|-1)\beta^{1/2}+\beta^{1/4})$. The probability that a
query $\q_h$ in $\Q$ is not well-separated is therefore bounded by
\begin{eqnarray*}
\frac{2}{\beta^{1/4}}+(1-\frac{2}{\beta^{1/4}})\sum_{k \neq j}
\frac{2}{(|k-h|-1)\beta^{1/2}+\beta^{1/4}}
&\leq& \frac{10}{\beta^{1/4}}+\sum_{k \neq
  j} \frac{2}{|k-h|\beta^{1/2}}=O\left(\frac{1}{\beta^{1/4}}+\frac{\lg n}{\beta^{1/2}}\right).
\end{eqnarray*}
Since $\beta=(wt_u)^9=\omega(\lg^2 n)$ this probability is $o(1)$,
and the result now follows from linearity of expectation
and Markov's inequality.
\end{proof}

Now let $S_i(Q,U) \subseteq S_i(U)$ denote the subset of cells in
$S_i(U)$ probed by the query algorithm of the claimed data structure
when answering all queries in a set of queries $Q$ after the sequence
of updates $U$ (i.e. the union of the cells probed for each query in
$Q$). Since a uniform random query from $\Q$ is uniform in $[n] \times
[n]$, we get by linearity of expectation that $\E[|S_i(\Q,U)|]=\beta^{i-1}t_i(U)$. From this, Lemma~\ref{lem:qsize}, Markov's
inequality and a union bound, we conclude
\begin{lemma}
\label{lem:well}
The query set $\Q$ is both well-separated and $|S_i(\Q,U)| \leq 4
\beta^{i-1} t_i(U)$ with probability at least $1/2$.
\end{lemma}
With this established, we are now ready to give an impossible
encoding of  $\Q$.

\paragraph{Encoding.}
In the following we describe the encoding procedure. The encoder
receives as input the set of queries $\Q$. He then executes the
following procedure:
\begin{enumerate}
\item The encoder first executes the fixed sequence of updates $U$ on
  the claimed data structure, and from this obtains the sets
  $S_{\lg_\beta n}(U),\dots,S_1(U)$. He then runs the query algorithm
  for every query $q \in \Q$ and collects the set $S_i(\Q,U)$.
\item If $\Q$ is not well-separated or if $|S_i(\Q,U)|>4\beta^{i-1}
  t_i(U)$, then the encoder sends a $1$-bit followed by a
  straightforward encoding of $\Q$ using $H(\Q)+O(1)$ bits in
  total. This is the complete encoding procedure when either $\Q$ is
  not well-separated or $|S_i(\Q,U)|>4\beta^{i-1} t_i(U)$.

\item If $\Q$ is both well-separated and $|S_i(\Q,U)|\leq 4
  \beta^{i-1}t_i(U)$, then the encoder first writes a
  $0$-bit and then executes the remaining four steps.

\item The encoder examines $\Q$ and finds the at most $\tfrac{1}{2}|\Q|$
  queries that are not well-separated. Denote this set $\Q'$. The
  encoder now writes down $\Q'$ by first specifying $|\Q'|$, then
  which vertical slabs contain the queries in $\Q'$ and finally what
  the coordinates of each query in $\Q'$ is within its slab. This
  takes
  $O(w)+\lg\binom{|\Q|}{|\Q'|}+|\Q'|\lg(n^2/\beta^{i-1})=O(w)+O(\beta^{i-1})+|\Q'|\lg(n^2/\beta^{i-1})$
  bits.

\item The encoder now writes down the cell set $S_i(\Q,U)$, including
  \textbf{only} the addresses and \textbf{not} the contents. This
  takes $o(H(\Q))$ bits since 
\begin{eqnarray*}
\lg \binom{|S_i(U)|}{|S_i(\Q,U)|} &=& O(\beta^{i-1} t_i(U) \lg(\beta t_u))\\&=&o(\beta^{i-1}
  \lg(n^2/\beta^{i-1})),
\end{eqnarray*}
where in the first line we used that $|S_i(U)|\leq \beta^it_u$ and $|S_i(\Q,U)| \leq 4\beta^{i-1} t_i(U)$. The second line follows from the fact that
  $t_i(U)=o(\lg_\beta n)=o(\lg (n^2/\beta^{i-1})/\lg(\beta
  t_u))$ since $\beta = \omega(t_u)$.

\item Next we encode the $x$-coordinates of the well-separated queries
  in $\Q$. Since we have already encoded which vertical slabs contain
  well-separated queries (we really encoded the slabs containing
  queries that are not well-separated, but this is equivalent), we do
  this by specifying only the offset within each slab. This takes
  $(|\Q|-|\Q'|)\lg(n/\beta^{i-1})+O(1)$ bits. Following that, the
  encoder considers the last grid $G_{2i-2}$, and for each
  well-separated query $q$, he writes down the $y$-offset of $q$
  within the grid cell of $G_{2i-2}$ hit by $q$. Since the grid cells
  of $G_{2i-2}$ have height $n/\beta^{i-1}$, this takes
  $(|\Q|-|\Q'|)\lg(n/\beta^{i-1})+O(1)$ bits. Combined with the encoding
  of the $x$-coordinates, this step adds a total of
  $(|\Q|-|\Q'|)\lg(n^2/\beta^{2i-2})+O(1)$ bits to the size of the
  encoding.

\item In the last step of the encoding procedure, the encoder
  simulates the query algorithm for every query in $[n] \times [n]$
  and from this obtains the set $Q_{S_i(\Q,U)}$, i.e. the set of all
  those queries that probe no cells in $S_i(U) \setminus
  S_i(\Q,U)$. Observe that $\Q \subseteq Q_{S_i(\Q,U)}$. The encoder now
  considers each of the grids $G_j$, for $j=2,\dots,2i-2$, and
  determines both the set of grid cells $G_j^{Q_{S_i(\Q,U)}} \subseteq
  G_j$ hit by a query in $Q_{S_i(\Q,U)}$, and the set of grid cells
  $G_j^{\Q} \subseteq G_j^{Q_{S_i(\Q,U)}} \subseteq G_j$ hit by a
  well-separated query in $\Q$. The last step of the encoding consists
  of specifying $G_j^{\Q}$. This is done by encoding which subset of
  $G_j^{Q_{S_i(\Q,U)}}$ corresponds to $G_j^{\Q}$. This takes $\lg
  \binom{|G_j^{Q_{S_i(\Q,U)}}|}{|G_j^{\Q}|}$ bits for each
  $j=2,\dots,2i-2$.

Since $|S_i(\Q,U)|=o(\beta^{i-1}\lg_\beta n)=o(\beta^{i-1}w)$ we get
from our contradictory assumption that the hitting number of
$Q_{S_i(\Q,U)}$ on each grid $G_j$ is $o(\beta^{i-3/4})$, thus
$|G_j^{Q_{S_i(\Q,U)}}|=o(\beta^{i-3/4})$. Therefore the above amount of
bits is at most
\begin{eqnarray*}
(|\Q|-|\Q'|)\lg(\beta^{i-3/4} e/(|\Q|-|\Q'|))(2i-3) &\leq \\
(|\Q|-|\Q'|)\lg(\beta^{1/4})2i+O(\beta^{i-1}i) &\leq \\
(|\Q|-|\Q'|)\tfrac{1}{4}\lg(\beta)2\lg_\beta n+O(\beta^{i-1}\lg_\beta
n)&\leq \\
(|\Q|-|\Q'|)\tfrac{1}{2}\lg n+o(H(\Q)).
\end{eqnarray*}

This completes the encoding procedure, and the encoder finishes by
sending the constructed message to the decoder.
\end{enumerate}

Before analysing the size of the encoding, we show that the decoder
can recover $\Q$ from the encoding.

\paragraph{Decoding.} 
In this paragraph we describe the decoding procedure. The decoder only
knows the fixed sequence $U=U_{\lg_\beta n},\dots,U_1$ and the message
received from the encoder. The goal is to recover $\Q$, which is done
by the following steps:
\begin{enumerate}
\item The decoder examines the first bit of the message. If this is a
  $1$-bit, the decoder immediately recovers $\Q$ from the remaining
  part of the encoding.

\item If the first bit is $0$, the decoder proceeds with this step and
  all of the below steps. The decoder executes the updates $U$ on the
  claimed data structure and obtains the sets $S_{\lg_\beta
    n}(U),\dots,S_1(U)$. From step 4 of the encoding procedure, the
  decoder also recovers $\Q'$.

\item From step 5 of the encoding procedure, the decoder now recovers
  the addresses of the cells in $S_i(\Q,U)$. Since the decoder has the
  data structure, he already knows the contents. Following this, the
  decoder now simulates every query in $[n] \times [n]$, and from this
  and $S_i(\Q,U)$ recovers the set $Q_{S_i(\Q,U)}$.

\item From step 6 of the encoding procedure, the decoder now recovers
  the $x$-coordinates of every well-separated query in $\Q$ (the
  offsets are enough since the decoder knows which vertical slabs
  contain queries in $\Q'$, and thus also those that contain
  well-separated queries). Following that, the decoder also recovers
  the $y$-offset of each well-separated query $q \in \Q$ within the
  grid cell of $G_{2i-2}$ hit by $q$ (note that the decoder does not
  know what grid cell it is, he only knows the offset).

\item From the set $Q_{S_i(\Q,U)}$ the decoder now recovers the set
  $G_j^{Q_{S_i(\Q,U)}}$ for each $j=2,\dots,2i-2$. This information is
  immediate from the set $Q_{S_i(\Q,U)}$. From $G_j^{Q_{S_i(\Q,U)}}$ and
  step 7 of the encoding procedure, the decoder now recovers $G_j^{\Q}$
  for each $j$. In grid $G_2$, we know that $\Q$ has only one query in
  every column, thus the decoder can determine uniquely from $G_2^{\Q}$
  which grid cell of $G_2$ is hit by each well-separated query in
  $\Q$. Now observe that the axis-aligned rectangle enclosing all
  $\beta^{1/2}$ grid cells in $G_{j+1}$ that intersects a fixed
  grid cell in $G_j$ has area $n^2/\beta^{i-1/2}$. Since we are
  considering well-separated queries, i.e. queries where no two lie
  within an axis-aligned rectangle of area $n^2/\beta^{i-1/2}$, this
  means that $G_{j+1}^{\Q}$ contains at most one grid cell in such a
  group of $\beta^{1/2}$ grid cells. Thus if $q$ is a well-separated
  query in $\Q$, we can determine uniquely which grid cell of $G_{j+1}$
  that is hit by $q$, directly from $G_{j+1}^{\Q}$ and the grid cell in
  $G_j$ hit by $q$. But we already know this information for grid
  $G_2$, thus we can recover this information for grid
  $G_3,G_4,\dots,G_{2i-2}$. Thus we know for each well-separated query
  in $\Q$ which grid cell of $G_{2i-2}$ it hits. From the encoding of
  the $x$-coordinates and the $y$-offsets, the decoder have thus
  recovered $\Q$.
\end{enumerate}

\paragraph{Analysis.}
Finally we analyse the size of the encoding. First consider the case
where $\Q$ is both well-separated and $|S_i(\Q,U)|\leq 4
\beta^{i-1}t_i(U)$. In this setting, the size of the message is
bounded by
\[
|\Q'|\lg(n^2/\beta^{i-1})+(|\Q|-|\Q'|)(\lg(n^2/\beta^{2i-2})+\tfrac{1}{2}\lg
n)+o(H(\Q))
\]
bits. This equals
\[
|\Q|\lg(n^{2+1/2}/\beta^{2i-2})+|\Q'|\lg(\beta^{i-1}/n^{1/2})+o(H(\Q))
\]
bits. Since we are considering an epoch $i \geq
\tfrac{15}{16}\lg_\beta n$, we have $\lg(n^{2+1/2}/\beta^{2i-2}) \leq
\lg(n^{5/8}\beta^2)$, thus the above amount of bits is upper bounded
by
\[
|\Q|\lg(n^{5/8}\beta^2)+|\Q'|\lg(n^{1/2})+o(H(\Q)).
\]
Since $|\Q'| \leq \tfrac{1}{2}|\Q|$, this is
again bounded by
\[
|\Q|\lg(n^{7/8}\beta^2)+o(H(\Q))
\]
bits. But $H(\Q) = |\Q|\lg(n^2/\beta^i) \geq |\Q|\lg n$, i.e. our
encoding uses less than $\tfrac{15}{16}H(\Q)$ bits.

Finally, let $E$ denote the event that $\Q$ is well-separated and at
the same time $|S_i(\Q,U)|\leq 4 \beta^{i-1} t_i(U)$, then the
expected number of bits used by the entire encoding is bounded by
\[
O(1)+\Pr[E](1-\Omega(1))H(\Q)+(1-\Pr[E])H(\Q)
\]
The contradiction is now reached by invoking Lemma~\ref{lem:well} to
conclude that $\Pr[E]\geq 1/2$.

\section{Concluding Remarks}
\label{sec:conclude}
In this paper we presented a new technique for proving dynamic cell
probe lower bounds. With this technique we proved the highest dynamic
cell probe lower bound to date under the most natural setting of cell
size $w=\Theta(\lg n)$, namely a lower bound of $t_q=\Omega((\lg
n/\lg(wt_u))^2)$.

While our results have taken the field of cell probe lower bounds one
step further, there is still a long way to go. Amongst the results
that seems within grasp, we find it a very intriguing open problem to
prove an $\omega(\lg n)$ lower bound for a problem where the queries
have a one bit output. Our technique crucially relies on the output
having more bits than it takes to describe a query, since otherwise
the encoder cannot afford to tell the decoder which queries to
simulate. Since many interesting data structure problems have a one
bit output size, finding a technique for handling this case would
allow us to attack many more fundamental data structure problems. As a
technical remark, we note that when proving static lower bounds using
the cell sampling idea, the encoder does not have to write down the
queries to simulate. This is because queries are completely solved
from the cell sample and need not read any other cells. Hence the
decoder can simply try to simulate the query algorithm for every
possible query and simply discard those that read cells outside the
sample. In the dynamic case, we still have to read cells associated to
other epochs. For the future epochs (small epochs), this is not an
issue since we know all these cells. However, when simulating the
query algorithm for a query that is not resolved by the sample,
i.e. it reads other cells from the epoch we are deriving a
contradiction for, we cannot recognize that the query fails. Instead,
we will end up using the cell contents written in past epochs and
could potentially obtain an incorrect answer for the query and we have
no way of recognizing this. We believe that finding a way to
circumvent the encoding of queries is the most promising direction for
improvements.

Applying our technique to other problems is also
an important task, however such problems must again have a logarithmic
number of bits in the output of queries.

\section{Acknowledgment}
The author wishes to thank Peter Bro Miltersen for much useful
discussion on both the results and writing of this paper.


\end{document}